\pdfoutput=1

\documentclass[letterpaper]{article}

\usepackage{amsmath,amssymb,amsthm}
\usepackage{enumerate}
\usepackage{graphicx}
\usepackage[margin=1in]{geometry}
\usepackage{colonequals}
\usepackage[numbers,sort]{natbib}
\usepackage{thm-restate}
\usepackage{mathtools}

\usepackage[
  pdftitle={Characterization of universal two-qubit Hamiltonians},
  pdfauthor={Andrew M. Childs, Debbie Leung, Laura Mancinska, Maris Ozols},
  bookmarksnumbered=true
]{hyperref}

\newcommand{\mc}[1]{\mathcal{#1}}

\newcommand{\tp}{^{\mathsf{T}}}
\newcommand{\ct}{^{\dagger}}
\newcommand{\co}{^{\ast}}
\newcommand{\defeq}{\colonequals}
\newcommand{\qefed}{\equalscolon}

\newcommand{\abs}[1]{\left|#1\right|}
\newcommand{\set}[1]{\left\{#1\right\}}
\newcommand{\norm}[1]{\left\|#1\right\|}
\newcommand{\opnorm}[1]{\left\|#1\right\|_\infty}
\newcommand{\bra}[1]{\left\langle#1\right|}
\newcommand{\ket}[1]{\left|#1\right\rangle}
\newcommand{\braket}[2]{\left\langle#1|#2\right\rangle}

\newcommand{\C}{\mathbb{C}}
\newcommand{\R}{\mathbb{R}}

\newcommand{\N}{\mathbb{N}}
\newcommand{\Q}{\mathbb{Q}}

\renewcommand{\O}[1]{\texorpdfstring{\ensuremath{\mathrm{O}(#1)}}{O(#1)}}
  \newcommand{\U}[1]{\texorpdfstring{\ensuremath{\mathrm{U}(#1)}}{U(#1)}}
  
  \newcommand{\SU}[1]{\texorpdfstring{\ensuremath{\mathrm{SU}(#1)}}{SU(#1)}}

\renewcommand{\u}[1]{\texorpdfstring{\ensuremath{\mathfrak{u}(#1)}}{u(#1)}}
  \newcommand{\su}[1]{\texorpdfstring{\ensuremath{\mathfrak{su}(#1)}}{su(#1)}}
  \newcommand{\M}[2]{\texorpdfstring{\ensuremath{{\mathrm{M}_{#1}}(#2)}}{M#1(#2)}}
  \newcommand{\SQ}[1]{\mc{S}_{#1}} 

\DeclareMathOperator{\tr}{Tr}
\DeclareMathOperator{\cl}{cl}
\DeclareMathOperator{\spn}{span}

\newcommand{\mx}[1]{\begin{pmatrix}#1\end{pmatrix}}                          
\newcommand{\smx}[1]{\bigl(\begin{smallmatrix}#1\end{smallmatrix}\bigr)}     
\newcommand{\sbmx}[1]{\left(\begin{smallmatrix}#1\end{smallmatrix}\right)}   

\newcommand{\universal}[1]{\texorpdfstring{\mbox{$#1$-universal}}{#1-universal}}
\newcommand{\nonuniversal}[1]{\texorpdfstring{\mbox{$#1$-non-universal}}{#1-non-universal}}
\newcommand{\qubit}[1]{\texorpdfstring{\mbox{$#1$-qubit}}{#1-qubit}}
\newcommand{\Tsimilar}[1]{\texorpdfstring{\mbox{$T$-similar#1}}{T-similar#1}}
\newcommand{\Tbasis}{\texorpdfstring{\mbox{$T$-basis}}{T-basis}}
\newcommand{\uF}{\U{4}}
\newcommand{\eto}[1]{e^{-i #1}}
\newcommand{\Lie}[1]{\ensuremath{\mc{L}\left(#1\right)}}
\newcommand{\id}{I}
\newcommand{\kt}[1]{\ket{\tilde{#1}}}
\def\be{\begin{equation}}
\def\ee{\end{equation}}
\newcommand{\vtheta}{\vec{\theta}}
\newcommand{\vphi}{\vec{\phi}}
\newcommand{\vq}{\vec{q}}
\newcommand{\thetas}{\theta_1, \theta_2, \theta_3, \theta_4}
\newcommand{\thetasn}{\theta_1, \dotsc, \theta_n}
\newcommand{\eps}{\varepsilon}

\theoremstyle{definition}
\newtheorem{theorem}{Theorem}
\newtheorem{lemma}{Lemma}
\newtheorem{claim}{Claim}
\newtheorem{definition}{Definition}
\newtheorem{cor}{Corollary}
\newtheorem{fact}{Fact}

\begin{document}


\title{Characterization of universal two-qubit Hamiltonians}
\author{Andrew M.\ Childs, Debbie Leung, Laura Mancinska, and Maris Ozols \\[5pt]
Department of Combinatorics \& Optimization \\ and Institute for Quantum Computing \\ University of Waterloo}
\date{}
\maketitle

\begin{abstract}
Suppose we can apply a given \qubit{2} Hamiltonian $H$ to any (ordered) pair of qubits. We say $H$ is \universal{n} if it can be used to approximate any unitary operation on $n$ qubits. While it is well known that almost any \qubit{2} Hamiltonian is \universal{2} (Deutsch, Barenco, Ekert 1995; Lloyd 1995), an explicit characterization of the set of \mbox{non-universal} \qubit{2} Hamiltonians has been elusive. Our main result is a complete characterization of \mbox{$2$-non-universal} \qubit{2} Hamiltonians. In particular, there are three ways that a \qubit{2} Hamiltonian $H$ can fail to be universal:
(1) $H$ shares an eigenvector with the gate that swaps two qubits,
(2) $H$ acts on the two qubits independently (in any of a certain family of bases), or
(3) $H$ has zero trace.
A \mbox{$2$-non-universal} \qubit{2} Hamiltonian can still be \universal{n} for some $n \geq 3$. We give some partial results on \universal{3}ity. Finally, we also show how our characterization of \universal{2} Hamiltonians implies the well-known result that almost any \qubit{2} unitary is universal.
\end{abstract}

\section{Introduction}

It is often useful to understand when a given set of resources is sufficient to perform universal computation. In particular, universal Hamiltonians have many applications in quantum computation.

Suppose we can implement one specific \qubit{2} Hamiltonian $H \in \u{4}$, where $\u{4}$ denotes the set of all $4 \times 4$ Hermitian matrices. Assume we have $n$ qubits and we can apply $H$ to any ordered pair of them for any amount of time. We say that $H$ is \universal{n} if it is possible to approximate any unitary evolution $U \in \U{2^n}$ to any desired accuracy by repeatedly applying $H$ to different pairs of qubits.

It is known that almost any \qubit{2} Hamiltonian is universal \cite{DBE,Lloyd}, i.e., non-universal \qubit{2} Hamiltonians form a measure-zero subset of $\u{4}$. Thus generic interactions are suitable for universal computation. But this does not address the issue of deciding whether a particular Hamiltonian is universal.

Given a specific $H \in \u{4}$, one can check numerically if $H$ is \universal{n} by determining whether $H$, when applied on different pairs of qubits, generates the Lie algebra of $\U{2^n}$ (see Section~\ref{sec:ProvingUniversality}). However, this characterization can be inconvenient for answering structural questions about universality. For example, suppose we can experimentally implement Hamiltonians of a certain restricted form, say, $\alpha (X \otimes I) + \beta (Y \otimes Y)$ for some $\alpha, \beta \in \R$. Determining which of these Hamiltonians are universal is not straightforward using the Lie-algebraic characterization. Indeed, until now there has been no simple \mbox{closed-form} characterization of the set of non-universal \qubit{2} Hamiltonians.

In this paper we characterize the set of all \nonuniversal{2} \qubit{2} Hamiltonians. In particular, our characterization easily answers questions such as those described above. We give a finite list of families of \nonuniversal{2} \qubit{2} Hamiltonians such that each family can be easily parametrized and together they cover all \nonuniversal{2} \qubit{2} Hamiltonians.

The remainder of the paper is organized as follows. Section~\ref{sec:Universality} introduces the concept of universality. We give our definition of universality, contrast this definition with some alternatives, review previous related work, and present a Lie-algebraic formulation. Section~\ref{sec:2Universal} then establishes our main result. We start from some simple families of Hamiltonians that are obviously \nonuniversal{2}, extend them with a class of operations that preserve this property, and then show that the extended families exactly characterize \universal{2}ity. Section~\ref{sec:3Universal} briefly summarizes what we know about \universal{3}ity. Finally, we conclude in Section~\ref{sec:Conclusion} with a discussion of some open problems.

In Appendix~\ref{app:UnitaryUniversality} (added after this paper was published), we show how our results easily imply the well-known result that almost any \qubit{2} unitary is universal.

\section{Universality in quantum computing}
\label{sec:Universality}

\subsection{Definition of universality}
\label{sec:Problem}

We begin with some basic definitions needed to precisely specify the problem addressed by this paper.

\begin{definition}
We say that $H$ is an \emph{\qubit{n} Hamiltonian} if $H \in \u{2^n}$, i.e., $H\in\M{2^n}{\C}$ ($\M{N}{\C}$ denotes the set of $N \times N$ complex matrices) and $H$ is Hermitian ($H\ct = H$).
\end{definition}

In this paper we mainly deal with \qubit{2} Hamiltonians, i.e., $4\times4$ Hermitian matrices. We often say ``a Hamiltonian $H$'' without explicitly mentioning that it is a \qubit{2} Hamiltonian.

\begin{definition}
\label{def:Simulate}
We say that we can \emph{simulate} a unitary transformation $U\in \U{N}$ using Hamiltonians $H_1,\dotsc,H_k\in\u{N}$ if for all $\varepsilon>0$ there exist $l \in \N$, $j_1, \dotsc, j_l \in \set{1,\dotsc,k}$, and $t_1, \dotsc, t_l > 0$ such that
\begin{equation}
  \opnorm{U-\eto{H_{j_1}t_1}\eto{H_{j_2}t_2}\dotso\eto{H_{j_l}t_l}} < \varepsilon.
\end{equation}
\end{definition}

Definition~\ref{def:Simulate} only allows the use of \mbox{positive} $t_i$ for simulating a unitary $U$ by Hamiltonians $H_1,\dotsc, H_k$, since $t_i$ corresponds to the length of time the system evolves according to $H_{j_i}$. However, this restriction can be relaxed to $t_i\in\R$. This is because an evolution by negative time can be approximated by evolving our system according to $H$ for some positive time instead (see Claim~\ref{cl:NegativeTime} in Appendix~\ref{app:NegativeTime} for a proof).

We only require the ability to \emph{approximate} any unitary to arbitrary precision. Such a definition is motivated by related universality problems based on discrete universal gate sets to be discussed below. We are not concerned about the time it takes to complete the simulation as long as we can simulate any unitary. Also, we do not assume the availability of ancillary systems.

\begin{definition}
\label{def:Universal}
Let $2\leq m\leq n$.
We say that an \qubit{m} Hamiltonian $H$ is \emph{\universal{n}} if we can simulate all unitary transformations in $\U{2^n}$ using Hamiltonians from the set
\begin{equation}
  \set{P (H \otimes I^{\otimes n-m}) P\ct \mid P \in \SQ{n}},
\end{equation}
where $\SQ{n}$ is the group of matrices that permute $n$ qubits. That is, we can apply $H$ to any ordered subset of $m$ qubits (out of $n$ qubits in total).
\end{definition}

The main goal of this paper is to characterize the set of \universal{2} \qubit{2} Hamiltonians. One motivation for this is that any \universal{2} \qubit{2} Hamiltonian is also \universal{n} for all integers $n\geq 2$ (see Lemma~\ref{lem:Universal} in Section~\ref{sec:ProvingUniversality}). Note that a \qubit{2} Hamiltonian $H$ is \universal{2} if we can simulate all unitary transformations in $\U{4}$ using $H$ and $THT$, where $T$ is the gate that swaps the two qubits, with the following representation in the computational (i.e., standard) basis:
\begin{equation}
  T\defeq\mx{1&0&0&0\\0&0&1&0\\0&1&0&0\\0&0&0&1}.
\label{eq:T1}
\end{equation}
To achieve our goal, we classify those \qubit{2} Hamiltonians that are \emph{not} \universal{2}.

\subsection{Other notions of universality}
\label{sec:AltDef}

Universal primitives for quantum computation, such as Hamiltonians and unitary gates, have been extensively studied previously; see for example Refs.\ \cite{DeutschQNetworks,DiVincenzo,Barenco,ReckZeilinger,ElementaryGates,Lloyd,DBE,KitaevUniversal}. Since the primitives are often physically motivated, there are different definitions of universality appropriate for different circumstances. First, one can study the universality of a set of quantum gates (instead of Hamiltonians). Second, one can study universality assuming ancillary qubits can be prepared and used to facilitate the computation. In particular, one might consider the following definitions of universality with ancillae:

\begin{definition}
\label{def:ReversibleUniversalAncilla}
For all $n,k\in\N$ let $\mc{C}(n,k)$ be the set of all functions from $n$-bit strings to $k$-bit strings. We say that a set of logical gates $\mc{S}$ is \emph{classically universal with ancillae} if for all $n,k\in\N$ and all $C\in\mc{C}(n,k)$ there exist $n_a\in\N$ and a logical circuit $G\in\mc{C}(n+n_a,k+n_a)$ containing gates exclusively from $\mc{S}$ that simulates $C$ using ancillae, i.e., there exists $a\in\set{0,1}^{n_a}$ such that for all $\psi\in\set{0,1}^n$ we have $(C(\psi),a)=G(\psi,a)$.
\end{definition}

\begin{definition}
\label{def:UniversalAncilla}
We say that a set of unitary gates $S$ is \emph{(quantumly) universal with ancillae} if for all $n\in\N$, all $\varepsilon>0$, and all $U\in\U{2^n}$, there exist $n_a\in\N$ and a quantum circuit $G\in\U{2^{n+n_a}}$ containing gates exclusively from $S$ that approximates $U$ with precision $\varepsilon$ using ancillae, i.e., there exists $a\in\set{0,1}^{n_a}$ such that for all $\ket{\psi}\in\C^{2^n}$ we have $\norm{(U \ket{\psi}) \otimes \ket{a} - G(\ket{\psi} \otimes \ket{a})} < \varepsilon$.
\end{definition}

Note that in the above definitions we assume the ability to prepare standard basis states. We allow initializing the ancillary bits to arbitrary standard basis states (as opposed to only $\ket{0}$) since some of the gates considered below (e.g., the Toffoli gate and Deutsch's gate) need ancillary bits prepared in basis states other than $\ket{0}$ to achieve universality. However, other reasonable definitions of universality with ancillae are possible.
(For example, the ancillary state need not be preserved in Definitions \ref{def:ReversibleUniversalAncilla} and \ref{def:UniversalAncilla}.)

In the classical case we can implement any logical gate exactly using elements of a universal gate set. In contrast, in the quantum case we only require the ability to \emph{approximate} any unitary to arbitrary precision. This definition is motivated by the need to use discrete universal gate sets to perform fault-tolerant quantum computing \cite{Boykin,Preskill}; such sets cannot implement a continuum of operations exactly.

\subsection{Previous results}
\label{sec:PreviousResults}

\subsubsection{Universal gate sets with ancillae} \label{sec:Ancillae}

It is well known that the gate set $\set{\text{NAND},\text{FANOUT}}$ is classically universal with ancillae. Deutsch \cite{DeutschQNetworks} showed that any gate from a certain family of \qubit{3} unitary gates is quantumly universal with ancillae. DiVincenzo \cite{DiVincenzo} suggested that it might be difficult to implement Deutsch's unitary gates as it is hard to build a mechanical device that brings three spins together. To obviate this, he devised a set of four \qubit{2} unitary gates that is quantumly universal with ancillae. Barenco \cite{Barenco} improved DiVincenzo's result by showing that a \emph{single} \qubit{2} unitary gate $A(\phi,\beta,\theta)$ is universal with ancillae, where
\begin{equation}
  A(\phi,\beta,\theta) \defeq
    \mx{1&0&0&0\\0&1&0&0\\
        0&0&e^{i\beta}\cos\theta&-ie^{i(\beta-\phi)}\sin\theta\\
        0&0&-ie^{i(\beta+\phi)}\sin\theta&e^{i\beta}\cos\theta}
\label{eq:Barenco}
\end{equation}
and $\phi$, $\beta$, and $\theta$ are irrational multiples of $\pi$ and of each other.

\subsubsection{Universal gate sets without ancillae}

Sets of unitary gates have also been found that can approximate any unitary transformation \emph{without} the use of ancillary qubits. It is well known \cite{NielsenChuang, Cybenko} that the Controlled-NOT gate together with all \qubit{1} gates form a universal gate set. Furthermore, several different finite sets of universal \qubit{2} quantum gates are known \cite{KitaevUniversal,Boykin,ShorFT}.

\subsubsection{Universality of a single \qubit{2} gate without ancillae}
\label{sec:PreviousResultsSingleGate}

In 1995, Deutsch, Barenco, and Ekert \cite{DBE} and Lloyd \cite{Lloyd} independently showed that almost any \qubit{2} gate can be used to approximate all \qubit{2} unitary evolutions. In other words, the set of non-universal unitary gates forms a measure-zero subset of the group $\U{4}$. Notably, in order to achieve universality, ancillary qubits are not required. The approaches used in \cite{DBE} and \cite{Lloyd} are similar in many respects and build upon the Lie-algebraic approach of DiVincenzo \cite{DiVincenzo}. Neither approach is constructive and both analyses revolve around the Lie algebra generated by $H$ and $THT$, where $H$ is a Hamiltonian corresponding to a generic unitary and $T$ is the gate exchanging the two qubits (recall equation~(\ref{eq:T1})). The proof in Ref.\ \cite{Lloyd} omits some details (some of which were later filled in by Weaver \cite{Weaver}), whereas Ref.\ \cite{DBE} provides a more complete proof.

Our work builds upon some of the techniques described in Ref.\ \cite{DBE}. Unfortunately, the arguments of that paper have some shortcomings:

\begin{enumerate}
\item The goal of \cite{DBE} is to establish the universality of a generic unitary $U\in\U{4}$. The argument begins by replacing $U$ with a ``Hamiltonian $H$ generating $U$,'' defined as a solution to $U=e^{iH}$. However, there can be different solutions generating different Lie algebras. As a simple example, both
\begin{equation}
  H \defeq\sbmx{  0 &0&0&0\\0&0&0&0\\0&0&0&0\\0&0&0&0} \quad \text{and} \quad
  H'\defeq\sbmx{2\pi&0&0&0\\0&0&0&0\\0&0&0&0\\0&0&0&0}
\end{equation}
generate $U = \id_4$, while only $H'$ can be used to approximate some non-identity evolutions. Thus one should give either a prescription for the choice of the generating Hamiltonian or a proof that different choices generically have the same power, but neither was provided in \cite{DBE}.

\item
The argument makes use of the fact that any gate $A$ given by (\ref{eq:Barenco}) is universal. However, such gates are only universal \emph{with ancillae} (because $\ket{00}$ is a fixed-point of both $A$ and $TAT$, so composing them cannot approximate any $U\in\U{4}$ that does not fix $\ket{00}$), yet the final result claims universality without the need for ancillae.

\item
The argument proceeds by considering a Hamiltonian $H_1$ that generates the gate $A$. The authors claim that $H_1$ is universal due to the linear independence of the following 16 nested commutators of $H_1$ and $TH_1T$:
\begin{equation}
\begin{aligned}
  H_1,   & \\
  H_2    &\defeq TH_1T,\\
  H_j    &\defeq i[H_1,H_{j-1}],\quad j\in\set{3,\dotsc,14},\\
  H_{15} &\defeq i[H_2,H_3],\\
  H_{16} &\defeq i[H_2,H_5],
\end{aligned}
\label{eq:Scheme}
\end{equation}

However, as in item 1, the claim may or may not hold depending on the choice of the Hamiltonian $H_1$ generating $A$. In fact, the most natural choice,
\begin{equation}
  H_1 \defeq \mx{0&0&0&0\\0&0&0&0\\0&0&\beta&-\theta e^{-i\phi}\\0&0&-\theta e^{i\phi}&\beta},
\end{equation}
does not generate $\u{4}$ since the entire Lie algebra fixes $\ket{00}$. However, there are other choices of $H_1$ for which $H_1,\dotsc,H_{16}$ are linearly independent. For example, if one chooses $H_1$ to act diagonally in a random basis on the degenerate $1$-eigenspace of $A$, with eigenvalues $2 \pi$ and $4 \pi$, then $H_1,\dotsc,H_{16}$ are found to be linearly independent in a numerical experiment.

For any explicit Hamiltonian $H$, it is simple to generate the $16$ matrices according to (\ref{eq:Scheme}) and their linear dependence is easily checked. If these $16$ matrices are linearly independent, then we say that (\ref{eq:Scheme}) certifies the universality of $H$.

\item
To show that almost any unitary gate is universal, non-universal gates are argued to lie in a submanifold of $\U{4}$ of at most 15 dimensions. The argument begins by considering a one-parameter family of Hamiltonians $H(k) = H + k (\tilde{H}-H)$ where $k\in\R$, $H$ is arbitrary, and $\tilde{H}$ is a fixed Hamiltonian whose universality is certified by (\ref{eq:Scheme}). Then, unless $k$ is a root of a certain polynomial of finite degree, (\ref{eq:Scheme}) also certifies the universality of $H(k)$. This argument is claimed to extend to a 16-dimensional neighborhood of $H$ (which could be parametrized as $H(k_1,\dotsc,k_{16})= H + k_1 (\tilde{H}^{(1)}-H) + \cdots + k_{16} (\tilde{H}^{(16)}-H)$). However, the explicit analysis of the relevant multivariate polynomial is omitted. Furthermore, the argument requires that (\ref{eq:Scheme}) certifies the universality of each of $\tilde{H}^{(1)},\dotsc,\tilde{H}^{(16)}$, but this is not demonstrated, and it is unclear to us whether it actually holds for some choice of $H_1$.

\end{enumerate}

\noindent Reference \cite{DBE} also conjectures that a \qubit{2} unitary gate is non-universal if and only if it
\begin{enumerate}
\item\label{item:permute} permutes states of some orthonormal basis or
\item\label{item:local} is a tensor product of single-qubit unitary gates.
\end{enumerate}
We note that a unitary gate $U$ satisfying item~\ref{item:permute} need not be non-universal, because $U$ and $TUT$ may not permute the same basis. We presume that the authors of~\cite{DBE} intended to require that \emph{both} $U$ and $TUT$ permute states of the same orthonormal basis.

In Theorem~\ref{thm:2NonUniversal} of this paper, we disprove the above conjecture and give a complete characterization of the set of non-universal \qubit{2} Hamiltonians, thereby resolving a variant of the above question.

\subsection{Proving universality}
\label{sec:ProvingUniversality}

The first step in our quest for a simple closed-form characterization of universal Hamiltonians is a characterization of universality in terms of Lie algebras, just as in \cite{DBE,Lloyd}.

\begin{definition}
\label{def:LieAlgebra}
We write \Lie{H_1,\dotsc,H_k} to denote the \emph{Lie algebra generated by Hamiltonians $H_1,\dotsc,H_k$}. It is defined inductively by the following three rules:
\begin{enumerate}
  \item $H_1,\dotsc,H_k \in \Lie{H_1,\dotsc,H_k}$,
  \item if $A,B \in \Lie{H_1,\dotsc,H_k}$ then $\alpha A + \beta B \in \Lie{H_1,\dotsc,H_k}$ for all $\alpha, \beta \in \R$, and
  \item if $A,B \in \Lie{H_1,\dotsc,H_k}$ then $i[A,B] \defeq i(AB-BA) \in \Lie{H_1,\dotsc,H_k}$.
\end{enumerate}
\end{definition}

The set of evolutions that can be simulated using a set of Hamiltonians is given by the following lemma:

\begin{lemma}
\label{lem:Lie}
Assume that we can evolve according to Hamiltonians $H_1, \dotsc, H_k$ for any desired amount of time. Then we can simulate the unitary $U$ if and only if
\begin{equation}
  U\in\cl\set{\eto{L}: L\in\Lie{H_1,\dotsc,H_k}},
\end{equation}
where ``$\cl$'' denotes the closure of a set.\footnote{This is false without the closure. For example, consider $H = \smx{1&0\\0&\sqrt{2}}$. We can use $H$ to simulate any diagonal $2 \times 2$ unitary but there are diagonal unitary matrices such as $\smx{1&0\\0&-1}$ that are not of the form $\eto{Ht}$ for $t \in \R$.}
\end{lemma}
One can easily prove the above lemma using the Lie product formula, the analogous formula for $e^{[A,B]}$, and the Campbell-Baker-Hausdorff formula~\cite{BCH}.

Now we can obtain a simpler and more practical sufficient condition for $n$-universa\-lity than the original one from Definition~\ref{def:Universal}.

\begin{cor}
\label{cor:UniversalLie}
Let $m\leq n$. Then an \qubit{m} Hamiltonian $H$ is \universal{n} if
\begin{equation}
  \Lie{\set{P (H \otimes I^{\otimes n-m}) P\ct : P \in \SQ{n}}} = \u{2^n},
\end{equation}
where $\SQ{n}$ is the group of matrices that permute $n$ qubits and $\u{2^n}$ is the set of all $2^n\times 2^n$ Hermitian matrices. In particular, a \qubit{2} Hamiltonian $H$ is \universal{2} if $\Lie{H,THT}=\u{4}$, where $\u{4}$ is the set of all $4\times 4$ Hermitian matrices.
\end{cor}

Now we proceed to show that if a Hamiltonian $H$ is \universal{n} then it is also \universal{n'} for all $n'\geq n$. Note that this is not completely trivial, since the added qubits are not ancillary, i.e., we have to be able to simulate any unitary on \emph{all} of the qubits.

\begin{lemma}
\label{lem:Universal}
If a Hamiltonian $H$ is \universal{n} for some $n \geq 2$, then it is also \universal{n'} for all $n' \geq n$. In particular, a \universal{2} \qubit{2} Hamiltonian $H$ is also \universal{n} for all integers $n\geq 2$.
\end{lemma}
\begin{proof} Since $H$ is \universal{n} for some $n \geq 2$, it can be used to simulate all unitary transformations in $\U{2^{n'}}$ that act non-trivially on no more than two qubits. But any unitary gate on $n'$ qubits can be decomposed into gates that act \mbox{non-trivially} only on one or two qubits without the need for ancillae \cite{ElementaryGates,NielsenChuang}, so $H$ is \universal{n'}.
\end{proof}

\section{Characterization of \universal{2} Hamiltonians}
\label{sec:2Universal}

In this section we classify the set of \qubit{2} Hamiltonians that are \emph{not} \universal{2}. Since we only consider \universal{2}ity, we simply say that a Hamiltonian is universal (instead of ``\universal{2}'') or \mbox{non-universal} (instead of ``not \universal{2}'') for the remainder of this section.

Our analysis relies on an equivalence relation that partitions the set of all \qubit{2} Hamiltonians into equivalence classes, each containing only universal or non-universal Hamiltonians (but not both). First we identify three families of \mbox{non-universal} Hamiltonians and extend each family to the union of the equivalence classes containing its family members. Then we show that each subset contains a special element whose universality (or non-universality) can be succinctly characterized. This allows us to show universality of any Hamiltonian not belonging to any of the three generalized non-universal families.

\subsection{The \texorpdfstring{$T$}{T} gate and the \Tbasis{}}
\label{sec:T}

The gate $T$ that swaps two qubits is of central importance since it is the only \mbox{non-trivial} permutation of two qubits. Recall that its matrix representation in the computational basis is
\begin{equation}
  T \defeq \mx{
         1 & 0 & 0 & 0 \\
         0 & 0 & 1 & 0 \\
         0 & 1 & 0 & 0 \\
         0 & 0 & 0 & 1 }.
\label{eq:T}
\end{equation}
It has two eigenspaces, namely
\begin{equation}
  E_- \defeq\spn_{\C}\set{\ket{01}-\ket{10}} \quad \text{ and } \quad
  E_+ \defeq\spn_{\C}\set{\ket{00},\ket{01}+\ket{10},\ket{11}},
\label{eq:EigenspacesOfT}
\end{equation}
where $E_-$ corresponds to the eigenvalue $-1$ and $E_+$ to the eigenvalue $+1$. The normalized vector
\begin{equation}
  \ket{s}\defeq\frac{\ket{01}-\ket{10}}{\sqrt{2}}
  \label{eq:Singlet}
\end{equation}
that spans $E_-$ is called the \emph{singlet} state.

We prove the following basic facts about the $T$ gate in Appendix \ref{app:Tgate}:

\begin{restatable}{fact}{CommutesWithTFact}
The singlet $\ket{s}$ is an eigenvector of a normal matrix $N \in \M{4}{\C}$ if and only if $[N,T]=0$.
\label{fact:CommutesWithT}
\end{restatable}

\begin{restatable}{fact}{OrthogonalEigenvectorFact}
A normal matrix $N \in \M{4}{\C}$ has a common eigenvector with the $T$ gate if and only if it has an eigenvector orthogonal to $\ket{s}$.
\label{fact:OrthogonalEigenvector}
\end{restatable}

\begin{restatable}{fact}{SingletEigenvecFact}
Suppose $U\in\uF$ and $[U,T]=0$. Then the singlet state $\ket{s}$ is an eigenvector of both $U$ and $U\ct$.
\label{fact:SingletEigenvec}
\end{restatable}

We will use both the computational basis and one in which $T$ is diagonal, with the singlet state as the first basis vector. For definiteness, we choose
\begin{equation}
  U_T \defeq \frac{1}{\sqrt{2}}
       \mx{ 0 & 1 & -1 &  0 \\
            0 & 1 & 1  &  0 \\
            1 & 0 & 0  &  1 \\
            1 & 0 & 0  & -1 }
\label{eq:ut}
\end{equation}
to implement the basis change. We call the resulting basis the \emph{\Tbasis{}}. The $T$ gate and the singlet state become $\tilde{T} \defeq U_T T U_T\ct$ and $\ket{\tilde{s}} \defeq U_T\ket{s}$ given by
\begin{equation}
  \tilde{T}=\mx{-1&0&0&0\\0&1&0&0\\0&0&1&0\\0&0&0&1} \quad \text{ and } \quad
  \ket{\tilde{s}}=\mx{1\\0\\0\\0}.
\label{eq:Ttilde}
\end{equation}

\subsection{Three simple families of non-universal Hamiltonians}
\label{sec:Examples}

Three families of \mbox{non-universal} Hamiltonians are easily identified.
\begin{fact}
\label{fact:Hyp1}
A \mbox{two-qubit} Hamiltonian $H$ is \mbox{non-universal} if any of the following conditions holds:
\begin{enumerate}
  \item $H$ is a local Hamiltonian, i.e., $H=H_1\otimes I+I\otimes H_2$, for some \qubit{1} Hamiltonians $H_1,H_2$,
  \item $H$ shares an eigenvector with the $T$ gate, or
  \item $\tr(H)=0$.
\end{enumerate}
\end{fact}

In Section \ref{sec:extended} we extend these families to larger sets of \mbox{non-universal} Hamiltonians, so the above do not literally exhaust the set of \mbox{non-universal} Hamiltonians. However, we prove in Section~\ref{sec:Converse} that the extended families contain all \mbox{non-universal} Hamiltonians, so these three families do capture the essence of what makes a Hamiltonian \mbox{non-universal}.

\subsection{\Tsimilar{ity}}
\label{sec:TsimTransformations}

The following equivalence relation between Hamiltonians is central to our analysis:
\begin{definition}
We say that matrices $A$ and $B$ are \emph{\Tsimilar{}} if there exists a unitary matrix $P$ such that $B=PAP^\dagger$ and $\left[P,T\right]=0$.
\end{definition}

Conjugation by $P$ \emph{preserves universality}, i.e., it maps universal \qubit{2} Hamiltonians to universal Hamiltonians and non-universal \qubit{2} Hamiltonians to non-universal Hamiltonians. In particular:

\begin{theorem}
\label{thm:Tsim}
Let $A,B$ be \Tsimilar{} \qubit{2} Hamiltonians. Then $A$ is universal if and only if $B$ is.
\end{theorem}

\begin{proof}
Assume \qubit{2} Hamiltonians $A$ and $B$ are \Tsimilar{}. Then there is some $P\in\U{4}$ such that $B=PAP^\dagger$ and $\left[P,T\right]=0$. Suppose $A$ is universal. We want to show that $B$ is also universal. We have to show that using $B$ we can simulate any $U\in\U{4}$ with any desired precision $\varepsilon>0$. Since $A$ is universal, we can simulate $P\ct UP\in\U{4}$ with precision $\varepsilon$, i.e., there exists $n\in\N$ and $t_1,\dotsc,t_n\geq 0$ such that
\begin{equation}
  \opnorm{P\ct UP-\eto{A t_1}\eto{T A T t_2}\eto{A t_3}\dotso\eto{T A T t_n}}<\varepsilon.
\label{eq:SimulationA}
\end{equation}
Since $TP=PT$, $B=PAP\ct$ and $e^{VMV^\dag}=Ve^MV^\dag$ for all unitary $V$ and all matrices $M$, we have
\begin{equation}
  \eto{B t_1}\eto{T B T t_2}\eto{B t_3}\dotso\eto{T B T t_n}
  =P\eto{A t_1} \eto{T A T t_2} \eto{A t_3} \dotso \eto{T A T t_n}P\ct.
\label{eq:SimulationB}
\end{equation}
Combining (\ref{eq:SimulationA}) with (\ref{eq:SimulationB}) and noting that the spectral norm is invariant under unitary conjugation gives
\begin{equation}
  \opnorm{U-\eto{B t_1}\eto{T B T t_2}\eto{B t_3}\dotso\eto{T B T t_n}}
  < \varepsilon.
\end{equation}
Hence $\eto{B t_1}\eto{T B T t_2}\eto{B t_3}\dotso\eto{T B T t_n}$ is the desired simulation of $U$ with precision $\varepsilon$. We conclude that $B$ is universal.
\end{proof}

Thus \Tsimilar{ity} partitions the set of all \qubit{2} Hamiltonians into equivalence classes, each containing only universal or non-universal Hamiltonians.

\subsection{Three extended families of non-universal Hamiltonians}
\label{sec:extended}

In view of Theorem \ref{thm:Tsim}, each family of non-universal Hamiltonians in Fact~\ref{fact:Hyp1} can be extended to include Hamiltonians that are \Tsimilar{} to its elements. We now analyze each of these three extended families.

\begin{enumerate}
\item \Tsimilar{ity} transformations do not preserve the set of local Hamiltonians. For example, when
\begin{equation}
  H\defeq\mx{1&0\\0&0}\otimes I + I \otimes \mx{1&0\\0&0}
  \quad\text{and}\quad
  P\defeq\mx{\frac{1}{\sqrt{2}}&0&0&\frac{1}{\sqrt{2}}\\
          0                &0&1&0                 \\
          0                &1&0&0                 \\
         \frac{1}{\sqrt{2}}&0&0&-\frac{1}{\sqrt{2}}
        },
\end{equation}
$H$ is local and $P$ commutes with $T$, but
$PHP^\dagger
  = I\otimes I+\frac{1}{2}
    \left[ \smx{0 &1\\1&0} \otimes \smx{0& 1\\1&0}
         - \smx{0&-i\\i&0} \otimes \smx{0&-i\\i&0} \right]$
which is non-local. Thus the extended family is strictly larger.

\item \Tsimilar{ity} transformations preserve the property of sharing an eigenvector with the $T$ gate:
\begin{lemma}
\label{lem:Condition2}
The set of \mbox{two-qubit} Hamiltonians sharing an eigenvector with the $T$ gate is closed under conjugation by unitary transformations that commute with $T$.
\end{lemma}
\begin{proof}
Let $U$ satisfy $[U,T]=0$ and let $\ket{v}$ be the eigenvector shared by $H$ and the $T$ gate, i.e., $H\ket{v}=\lambda_H\ket{v}$ and $T\ket{v}=\lambda_T\ket{v}$ for some $\lambda_H$, $\lambda_T$. We claim that $U\ket{v}$ is an eigenvector shared by the $T$ gate and $UHU^\dagger$. First, note that $UHU^\dagger (U\ket{v})= UH\ket{v} = \lambda_H U\ket{v}$. We also have $T(U\ket{v})=UT\ket{v}=\lambda_T U\ket{v}$. Thus $U\ket{v}$ is an eigenvector shared by the $T$ gate and $UHU^\dagger$.
\end{proof}
Therefore, the extension does not add more non-universal Hamiltonians to this family.

\item The set of traceless Hamiltonians is preserved by \Tsimilar{ity} transformations.
\end{enumerate}

\noindent Using the above, we generalize Fact~\ref{fact:Hyp1} to
the following:
\begin{lemma}
\label{lem:Hyp2}
A \mbox{two-qubit} Hamiltonian $H$ is \mbox{non-universal} if any of the following conditions holds:
\begin{enumerate}
  \item $H$ is \Tsimilar{} to a local Hamiltonian,
  \item $H$ shares an eigenvector with the $T$ gate, or
  \item $\tr(H)=0$.
\end{enumerate}
\end{lemma}

Another easily recognized family of non-universal Hamiltonians is the set of generators of orthogonal transformations. However, this set can be shown to be contained in the first family of the above lemma \cite{Thesis}. Similarly, Hamiltonians with degenerate eigenvalues can be shown to be non-universal, since they always share an eigenvector with the $T$ gate (see Fact~\ref{fact:DegEigenvals} in Appendix~\ref{app:Tgate}).

It is straightforward to check whether a given Hamiltonian belongs to the last two families of non-universal Hamiltonians in Lemma \ref{lem:Hyp2}. The following lemma (proved in Appendix \ref{app:Tsimilarlocal}) gives an efficient method to check whether a given Hamiltonian with non-degenerate eigenvalues is \Tsimilar{} to a local Hamiltonian.
\begin{restatable}{lemma}{PatternLemma}
A \qubit{2} Hamiltonian is \Tsimilar{} to a local Hamiltonian if and only if it has an orthonormal basis of eigenvectors $\ket{v_1},\ket{v_2},\ket{v_3},\ket{v_4}$ corresponding to the eigenvalues $\lambda_1,\lambda_2,\lambda_3,\lambda_4$ so that
\begin{enumerate}
  \item $\abs{\braket{v_1}{s}}=\abs{\braket{v_2}{s}}$ and $\abs{\braket{v_3}{s}}=\abs{\braket{v_4}{s}}$, and
  \item $\lambda_1 + \lambda_2 = \lambda_3 + \lambda_4$,
\end{enumerate}
where $\ket{s}$ is the singlet state defined in equation (\ref{eq:Singlet}).
\label{lem:Pattern}
\end{restatable}

\subsection{The three extended families of non-universal Hamiltonians are exhaustive}
\label{sec:Converse}

In this section we show that the list of non-universal families of Hamiltonians in Lemma~\ref{lem:Hyp2} is in fact complete. This is done by analyzing a special member of each \Tsimilar{ity} equivalence class.

\subsubsection{Tridiagonal form}
\label{sec:Tridiagonal}
We now introduce a normal form for \qubit{2} Hamiltonians.

\begin{definition}
We say that a \qubit{2} Hamiltonian is in \emph{tridiagonal form} if it is of the form~
\begin{equation}
  \mx{
    a & b & 0 & 0 \\
    b & c & d & 0 \\
    0 & d & e & f \\
    0 & 0 & f & g },
\label{eq:Tridiagonal}
\end{equation}
where $a,b,c,d,e,f,g \in \R$ and $b,d,f \geq 0$. If either of $b,d$ is 0, we additionally require that
\begin{itemize}
\item if $b=0$, then $d=f=0$ and $c \geq e \geq g$, and
\item if $d=0$, then $f=0$ and $e \geq g$.
\end{itemize}
Note that a tridiagonal Hamiltonian is of one of the following types:\vspace{0.1in}

\begin{center}
\begin{tabular}{cccc}
  $\sbmx{
    * & + & 0 & 0 \\
    + & * & + & 0 \\
    0 & + & * & + \\
    0 & 0 & + & * }$
  &$\sbmx{
    * & + & 0 & 0 \\
    + & * & + & 0 \\
    0 & + & * & 0 \\
    0 & 0 & 0 & * }$
  &$\sbmx{
    * & + & 0   & 0 \\
    + & * & 0   & 0 \\
    0 & 0 & *_1 & 0 \\
    0 & 0 & 0   & *_2 }$
  &$\sbmx{
    *   & 0   & 0   & 0  \\
    0   & *_1 & 0   & 0  \\
    0   & 0   & *_2 & 0  \\
    0   & 0   & 0   & *_3}$\vspace{0.1in}\\
   Type 1 & Type 2 & Type 3 & Type 4
\end{tabular}\vspace{0.1in}
\end{center}
where $*_1\geq*_2\geq*_3$ and ``$+$'' stands for a positive entry and ``$*$'' for any real entry.
\end{definition}

When given a \qubit{2} Hamiltonian in tridiagonal form, we will often use the letters $a,b,c,d,e,f,g$ to refer to its entries as indicated in equation~(\ref{eq:Tridiagonal}).

\begin{definition}
\label{def:TridiagonalForm}
For any \qubit{2} Hamiltonian $H$, we say that \emph{$\Xi$ is a tridiagonal form of $H$} if $H$ and $\Xi$ are \Tsimilar{} and $\Xi$ is tridiagonal in the \Tbasis{}. (We will use $\tilde{\Xi} \defeq U_T \Xi U_T^\dag$ to denote $\Xi$ in the \Tbasis{}.)
\end{definition}

It follows from the definition that \Tsimilar{} \qubit{2} Hamiltonians share the same tridiagonal forms (if they exist). We now show that for every \qubit{2} Hamiltonian a tridiagonal form indeed exists and is in fact unique. Thus, each equivalence class is uniquely characterized by the tridiagonal form of its Hamiltonians.
\begin{lemma}
\label{thm:Tridiagonalization}
Every \qubit{2} Hamiltonian $H$ has a unique tridiagonal form $\Xi$.
\end{lemma}

\begin{proof}
Since \Tsimilar{ity} is basis-independent, we prove the lemma in the \Tbasis{}. In other words, we prove that $\tilde{H} \defeq U_T H U_T^\dag$ is \Tsimilar{} to a unique tridiagonal matrix. Note that in the \Tbasis{}, $T$-similar matrices are related by conjugation by some unitary $V \in \U{1}\oplus\U{3}$.

Let the first column of $\tilde{H}$ be $(h_1, h_2, h_3, h_4)\tp$, where $\norm{(h_2, h_3, h_4)\tp} = b \geq 0$. Then we can find $P_1 \in \id_1 \oplus\U{3}$ such that the first column of $\tilde{H}_1 \defeq P_1 \tilde{H} P_1\ct$ is $(h_1, b, 0, 0)\tp$. Now let the second column of $\tilde{H}_1$ be $(h_1', h_2', h_3', h_4')\tp$, where $\norm{(h_3', h_4')\tp} = d \geq 0$, and choose $P_2 \in \id_2 \oplus \U{2}$ such that the second column of $\tilde{H}_2 \defeq P_2 \tilde{H}_1 P_2\ct$ is $(h_1', h_2', d, 0)\tp$. Note that the first column of $\tilde{H}_2$ remains the same as for $\tilde{H}_1$. Finally, we can find $P_3 \in \id_3 \oplus \U{1}$ such that the last entry $f$ of the third column of $\tilde{H}_3 \defeq P_3 \tilde{H}_2 P_3\ct$ is real and \mbox{non-negative}. Since $\tilde{H}_3$ is Hermitian, its diagonal entries are real and it has the form (\ref{eq:Tridiagonal}). If neither $b$ nor $d$ is zero, we are done. If $b=0$, we diagonalize the lower right $3\times 3$ block of $\tilde{H}_3$ by conjugating with unitary transformations of the form $1\oplus\U{3}$. Similarly, if $d=0$ we diagonalize the lower right $2\times 2$ block. Thus we obtain a tridiagonal form of $H$.

Now we show that $\Xi$ is unique. If $\Xi_1$ and $\Xi_2$ are both tridiagonal forms of $H$, then in the \Tbasis{} $\tilde{\Xi}_1$ and $\tilde{\Xi}_2$ are related by conjugation by some $V \in \U{1}\oplus\U{3}$. We first consider $\tilde{\Xi}_1$ of type 1. Since the first column of $\tilde{\Xi}_2$ has to be of the form $(a,b,0,0)\tp$ for some $a,b\in\R$, $b>0$, $V$ has to be of the form $e^{i\varphi}\id_2\oplus\U{2}$ for some $\varphi\in\R$. Similarly, by considering the second and third columns of $\tilde{\Xi}_2$, we conclude that $V=e^{i\varphi}\id_4$. Thus, we have $\tilde{\Xi}_2 = (e^{i\varphi} \id_4) \tilde{\Xi}_1 (e^{-i\varphi} \id_4) = \tilde{\Xi}_1$. If $\tilde{\Xi}_1$ is of type 2, 3, or 4, similar reasoning can be applied; in each case, the form of $V$ is constrained so that $\tilde{\Xi}_1 = \tilde{\Xi}_2$.
\end{proof}

\subsubsection{Tridiagonal forms of non-universal Hamiltonians}
\label{sec:TridiagonalNonuniversal}

In this section we give a simple characterization of the three families of \mbox{non-universal} Hamiltonians listed in Lemma~\ref{lem:Hyp2} in terms of their tridiagonal forms.

\begin{lemma}
\label{lem:TridiagonalEigenvector}
Let $H$ be a \qubit{2} Hamiltonian and let $\Xi$ be its tridiagonal form, with $\tilde{\Xi}$ given by equation~(\ref{eq:Tridiagonal}). Then $H$ has a common eigenvector with the $T$ gate if and only if $bdf=0$.
\end{lemma}

\begin{proof}
By Fact~\ref{fact:OrthogonalEigenvector}, $H$ has a common eigenvector with $T$ if and only if $H$ has an eigenvector orthogonal to the singlet $\ket{s}$. By definition of the tridiagonal form, there is a unitary conjugating $H$ to $\tilde{\Xi}$, $T$ to $\tilde{T}$, and taking $\ket{s}$ to $\kt{s}$. Thus it suffices to show that $\tilde{\Xi}$ has an eigenvector orthogonal to $\kt{s}$ if and only if $bdf=0$.

If $bdf=0$, then $\tilde{\Xi}$ has an invariant subspace orthogonal to $\ket{\tilde{s}}$. This subspace has dimension $3$, $2$, or $1$ if $b=0$, $d=0$, or $f=0$, respectively. In any case, it contains at least one eigenvector, so $\tilde{\Xi}$ has an eigenvector orthogonal to $\kt{s}$.

If $\tilde{\Xi}$ has an eigenvector $\kt{v}$ that is orthogonal to $\kt{s}$, then $\kt{v} = (0, v_2, v_3, v_4)\tp$ for some $v_2,v_3,v_4\in\C$, not all zero. Since $\tilde{\Xi}$ is Hermitian, we have $\tilde{\Xi} \kt{v}=r\kt{v}$ for some $r\in \R$, or equivalently,
\begin{equation}
  \mx{b v_2\\c v_2+ d v_3\\d v_2+e v_3+ f v_4 \\ f v_3+ g v_4} = \mx{0\\r v_2\\r v_3\\r v_4}.
\end{equation}
From the first entry, $b v_2 = 0$ so that $b=0$ or $v_2=0$. If $b \neq 0$, then $v_2=0$ and from the second entry, $d v_3=0$, so either $d=0$ or $v_3=0$. If $d \neq 0$, then $v_3=0$ and from the third entry, $f v_4 = 0$ and so, if $f \neq 0$, $v_4 = 0$ which contradicts $\kt{v}\neq 0$. Thus, $bdf=0$.
\end{proof}

\begin{lemma}
\label{lem:TridiagonalLocal}
Let $H$ be a \qubit{2} Hamiltonian not sharing an eigenvector with $T$ and let $\Xi$ be its tridiagonal form. Then $H$ is \Tsimilar{} to a local Hamiltonian if and only if $a=c=e=g$ for $\tilde{\Xi}$ as given in equation~(\ref{eq:Tridiagonal}).
\end{lemma}

\begin{proof}
Assume $a=c=e=g$. $H$ is \Tsimilar{} to a local Hamiltonian if and only if $\Xi$ is, so we can apply Lemma~\ref{lem:Pattern} to $\Xi$. A straightforward calculation gives eigenvectors $\ket{v_{i,j}}$ of $\tilde{\Xi}$ for $i,j \in \{0,1\}$, with eigenvalues
\begin{equation}
  \lambda_{i,j} = a + (-1)^{i} \sqrt{\frac{b^2 + d^2 + f^2 + (-1)^{j} z}{2}} \quad
  \text{where}
  \quad
  z \defeq \sqrt{b^4 + d^4 + f^4 + 2 (b^2 d^2 + d^2 f^2 - b^2 f^2)}.
\end{equation}
The overlaps of these eigenvectors with the singlet state are
\begin{equation}
  \abs{\braket{v_{i,j}}{\tilde{s}}} = \sqrt{\frac{z + (-1)^j (b^2 - d^2 - f^2)}{4z}}.
\end{equation}
For each $j \in \{0,1\}$, $\lambda_{0,j} + \lambda_{1,j} = 2a$ and $\abs{\braket{v_{0,j}}{\tilde{s}}} = \abs{\braket{v_{1,j}}{\tilde{s}}}$, so both conditions in Lemma~\ref{lem:Pattern} are satisfied. Hence $H$ is \Tsimilar{} to a local Hamiltonian.

Now assume that $H$ is \Tsimilar{} to a local Hamiltonian. We prove that $\tilde{\Xi}$ has $a=c=e=g$ by explicitly computing it. We do this in two steps: first we show that $H$ is \Tsimilar{} to some $H'$ of the form
\begin{equation}
  H' = \alpha \id_4 + (x_1 X + z_1 Z) \otimes \id + \id \otimes (z_2 Z)
  \label{eq:SimplifiedHamiltonian}
\end{equation}
for some $\alpha,x_1,z_1,z_2\in\R$ (where $\id \defeq \smx{1&0\\0&1}$, $X \defeq \smx{0&1\\1&0}$, $Y \defeq \smx{0&-i\\i&0}$, $Z \defeq \smx{1&0\\0&-1}$ are Pauli matrices), and then we find the common tridiagonal form of $H$ and $H'$.

\paragraph{Step 1.} Consider conjugating $H$ by $U \otimes U$ where $U\in\SU{2}$ (clearly, $[U\otimes U,T]=0$). It suffices to consider a local Hamiltonian
\begin{equation}
  H = H_1 \otimes \id + \id \otimes H_2
  \label{eq:HamiltonianForTensorProduct}
\end{equation}
for some \qubit{1} Hamiltonians $H_1$ and $H_2$. Pick $V \in \SU{2}$ such that $V H_2 V\ct$ is diagonal. Then pick a diagonal matrix $D \in \SU{2}$ such that $D V H_1 V\ct D\ct$ is a real matrix. Note that $D V H_2 V\ct D\ct$ is still diagonal. Therefore, $U H U\ct$ is of the form (\ref{eq:SimplifiedHamiltonian}), where $U \defeq (D V) \otimes (D V)$.

\paragraph{Step 2.} Recall that $H'$ in the \Tbasis{} is given by $\tilde{H'} \defeq U_T H' U_T\ct$. Using equation (\ref{eq:SimplifiedHamiltonian}), we get
\begin{equation}
  \tilde{H'}
  = \alpha (I \otimes I)
    + x_1 (Y \otimes Y)
    + z_1 (I \otimes X)
    - z_2 (Z \otimes X)
  = \mx{ \alpha    & z_1 - z_2 & 0       & -x_1      \\
         z_1 - z_2 & \alpha  & x_1       & 0         \\
         0         & x_1     & \alpha    & z_1 + z_2 \\
         -x_1      & 0       & z_1 + z_2 & \alpha    }.
\end{equation}
Note that $x_1\neq 0$, or else, by Lemma~\ref{lem:TridiagonalEigenvector}, we contradict the fact that $H$ does not share an eigenvector with $T$. We apply one more \Tsimilar{ity} transformation to bring $\tilde{H'}$ into tridiagonal form. Let $l \defeq \sqrt{x_1^2 + (z_1 - z_2)^2} > 0$ and
\begin{equation}
  Q \defeq \mx{ 1 & 0                   & 0 & 0                   \\
                0 & \frac{z_1 - z_2}{l} & 0 & \frac{x_1}{l}       \\
                0 & 0                   & 1 & 0                   \\
                0 & \frac{-x_1}{l}      & 0 & \frac{z_1 - z_2}{l} }
\end{equation}
which is in $\id_1 \oplus \U{3}$. Then
\begin{equation}
  Q\ct \tilde{H'} Q =
    \mx{ \alpha & l                    & 0                               & 0                               \\
         l      & \alpha               & \frac{-2 x_1 z_2}{l}            & 0                               \\
         0      & \frac{-2 x_1 z_2}{l} & \alpha                          & \frac{x_1^2 + z_1^2 - z_2^2}{l} \\
         0      & 0                    & \frac{x_1^2 + z_1^2 - z_2^2}{l} & \alpha                          }.
\label{eq:Xi}
\end{equation}
Since $H$ does not share an eigenvector with $T$, the $(2,1)$, $(3,2)$, $(4,3)$ entries of $Q\ct \tilde{H'} Q$ are all nonzero. Their signs can be made positive by conjugation with a diagonal matrix with diagonal entries $\pm 1$, preserving the diagonal elements (note that this is a \Tsimilar{ity} transformation). This tridiagonal form of $H'$ and $H$ has equal diagonal entries as claimed.
\end{proof}

Using Lemmas~\ref{lem:TridiagonalEigenvector} and \ref{lem:TridiagonalLocal}, we can restate Lemma~\ref{lem:Hyp2} in terms of the tridiagonal form:
\begin{cor}
\label{cor:Hyp2}
Let $H$ be a \qubit{2} Hamiltonian with tridiagonal form $\Xi$, with $\tilde{\Xi}$ given by (\ref{eq:Tridiagonal}). Then $H$ is non-universal if
\begin{enumerate}
  \item $\tilde{\Xi}$ has $bdf=0$,
  \item $\tilde{\Xi}$ has $a=c=e=g$, or
  \item $\tilde{\Xi}$ has $a+c+e+g=0$.
\end{enumerate}
\end{cor}

\subsubsection{Universality certificate for tridiagonal Hamiltonians}
\label{sec:UniversalityCertificate}

Given a Hamiltonian $H$ that does not satisfy any of the conditions of Corollary~\ref{cor:Hyp2}, we provide a list of $16$ linearly independent linear combinations of nested commutators of $\tilde{\Xi}$ and $\tilde{T}\tilde{\Xi}\tilde{T}$. This shows that $\Lie{H,THT}=\mc{L}({\tilde{\Xi}, \tilde{T}\tilde{\Xi}\tilde{T}})=\u{4}$. Hence it follows from Corollary~\ref{cor:UniversalLie} that $H$ is universal.

Let $E_{k,l}\defeq\ket{k}\bra{l}$ and define a basis for $\su{4}$ (i.e., for traceless $4 \times 4$ Hermitian matrices) as follows:
\begin{align}
  X_{k,l} &\defeq E_{k,l} + E_{l,k},      & (1 \leq k < l \leq 4) \\
  Y_{k,l} &\defeq -i E_{k,l} + i E_{l,k}, & (1 \leq k < l \leq 4) \\
  Z_{k} &\defeq E_{k,k} - E_{k+1,k+1}.    & (1 \leq k \leq 3)
\end{align}
These $15$ matrices together with any Hermitian matrix with \mbox{non-zero} trace form a basis for $\u{4}$. We now obtain these basis vectors as nested commutators of $\tilde{\Xi}$ and $\tilde{T} \tilde{\Xi} \tilde{T}$.

By violation of the first condition in Corollary~\ref{cor:Hyp2}, $bdf \neq 0$. Thus we can generate $A \defeq \frac{1}{2b} i[\tilde{\Xi},\tilde{T}\tilde{\Xi}\tilde{T}]$ and
\begin{align}
  X_{1,2} &= \frac{1}{2b} (\tilde{\Xi} - \tilde{T}\tilde{\Xi}\tilde{T}), \\
  Y_{1,3} &= \frac{1}{3d} \bigl(i[i[X_{1,2}, A], X_{1,2}] - 4 A\bigr), \\
  X_{2,3} &= i[X_{1,2}, Y_{1,3}].
\end{align}
Next, we can generate $B \defeq \frac{1}{2} (\tilde{\Xi} + \tilde{T}\tilde{\Xi}\tilde{T})$.
To obtain $Y_{1,2}$ we consider three cases:
\begin{equation}
  Y_{1,2} =
  \begin{cases}
    \dfrac{1}{a-c} (d Y_{1,3} + A)                                      & \text{ if } a \neq c, \\
    \dfrac{1}{c-e} \, i[Y_{1,3}, i[B, X_{2,3}]]                         & \text{ if } c \neq e, \\
    \dfrac{1}{a-g} \dfrac{1}{f^2} \, i \left[
   i[X_{2,3}, B], i[B, i[Y_{1,3}, B]] \right] & \text{ otherwise ($a = c = e \neq g$)}.
  \end{cases}
\end{equation}
One of these cases has to hold since the second condition in Corollary~\ref{cor:Hyp2} is violated. We next obtain
\begin{align}
  X_{1,3} &= i[Y_{1,2},X_{2,3}], \\
  X_{1,4} &= \frac{1}{f} \bigl((c-e)X_{1,3} + i[A, X_{2,3}] + i[Y_{1,3}, B]\bigr).
\end{align}
We obtain the remaining basis elements as follows:
\begin{align}
  X_{2,4} &= i[X_{1,4}, Y_{1,2}], &
  X_{3,4} &= i[X_{1,4}, Y_{1,3}], &
  Y_{1,4} &= i[X_{2,4}, X_{1,2}], \\
  Y_{2,3} &= i[X_{1,3}, X_{1,2}], &
  Y_{2,4} &= i[X_{1,4}, X_{1,2}], &
  Y_{3,4} &= i[X_{1,4}, X_{1,3}], \\
  Z_1 &= \frac{1}{2} \, i[Y_{1,2}, X_{1,2}], &
  Z_2 &= \frac{1}{2} \, i[Y_{2,3}, X_{2,3}], &
  Z_3 &= \frac{1}{2} \, i[Y_{3,4}, X_{3,4}].
\end{align}
At this point we can generate \su{4}. If the third condition in Corollary~\ref{cor:Hyp2} does not hold, then $\tr(\tilde{\Xi}) \neq 0$, so adding $\tilde{\Xi}$ gives all of \u{4}.

\subsubsection{Complete classification of \universal{2} \qubit{2} Hamiltonians}
\label{sec:result}

Combining the result of the previous section with Corollary~\ref{cor:Hyp2} gives the following theorem.

\begin{theorem}
\label{thm:2NonUniversal}
A \mbox{two-qubit} Hamiltonian $H$ is \universal{2} if and only if it does not satisfy any of the following conditions:
\begin{enumerate}
  \item $H$ is \Tsimilar{} to a local Hamiltonian,
  \item $H$ shares an eigenvector with $T$, the gate that swaps two qubits, or
  \item $\tr(H)=0$.
\end{enumerate}
\end{theorem}

These conditions are easy to check by computing $\tilde{\Xi}$ and applying Corollary~\ref{cor:Hyp2}.

\section{\nonuniversal{3} Hamiltonians}
\label{sec:3Universal}

It turns out that there are \nonuniversal{2} \qubit{2} Hamiltonians that are nevertheless \universal{3}. In fact, numerical evidence suggests that almost any traceless \nonuniversal{2} Hamiltonian is \universal{3}. We do not know a complete characterization of \universal{3} \qubit{2} Hamiltonians. However, the following are sufficient conditions for a \qubit{2} Hamiltonian to be \nonuniversal{3}.

\begin{lemma}
\label{lem:3NonUniversal}
A \qubit{2} Hamiltonian $H$ is \nonuniversal{3} if any of the following conditions holds:
\begin{enumerate}
  \item $H$ is a local Hamiltonian,
  \item $H$ has an eigenvector of the form $\ket{a}\!\ket{a}$ for some $\ket{a}\in\C^2$,
  \item $\tr(H)=0$,
  \item $H=r\id_4+(U\otimes U) A (U\otimes U)\ct$ for some $r\in\R$, $U\in\U{2}$, and some antisymmetric Hamiltonian $A\in\u{4}$ ($A$ is antisymmetric if $A\tp=-A$),
  \item $[H,U\otimes U]=0$ for some $U\in\U{2}$ with distinct eigenvalues.
\end{enumerate}
\end{lemma}

It is easy to see that these Hamiltonians are indeed \nonuniversal{3}. In fact, they are also \nonuniversal{n} for all $n \geq 3$. Therefore, if one could show that this list is complete, then it would provide a complete characterization of \universal{n} \qubit{2} Hamiltonians.

Recall that a \nonuniversal{3} \qubit{2} Hamiltonian is also \nonuniversal{2} (see Lemma~\ref{lem:Universal}). For each family $\mc{F}_3$ of \nonuniversal{3} Hamiltonians in Lemma~\ref{lem:3NonUniversal} there is a family $\mc{F}_2$ of \nonuniversal{2} Hamiltonians from Theorem~\ref{thm:2NonUniversal} sucht that $\mc{F}_3\subseteq\mc{F}_2$ (see Figure~\ref{fig:Diagram}).

\begin{figure}[ht]
  \centering
  \includegraphics[width=.75\textwidth]{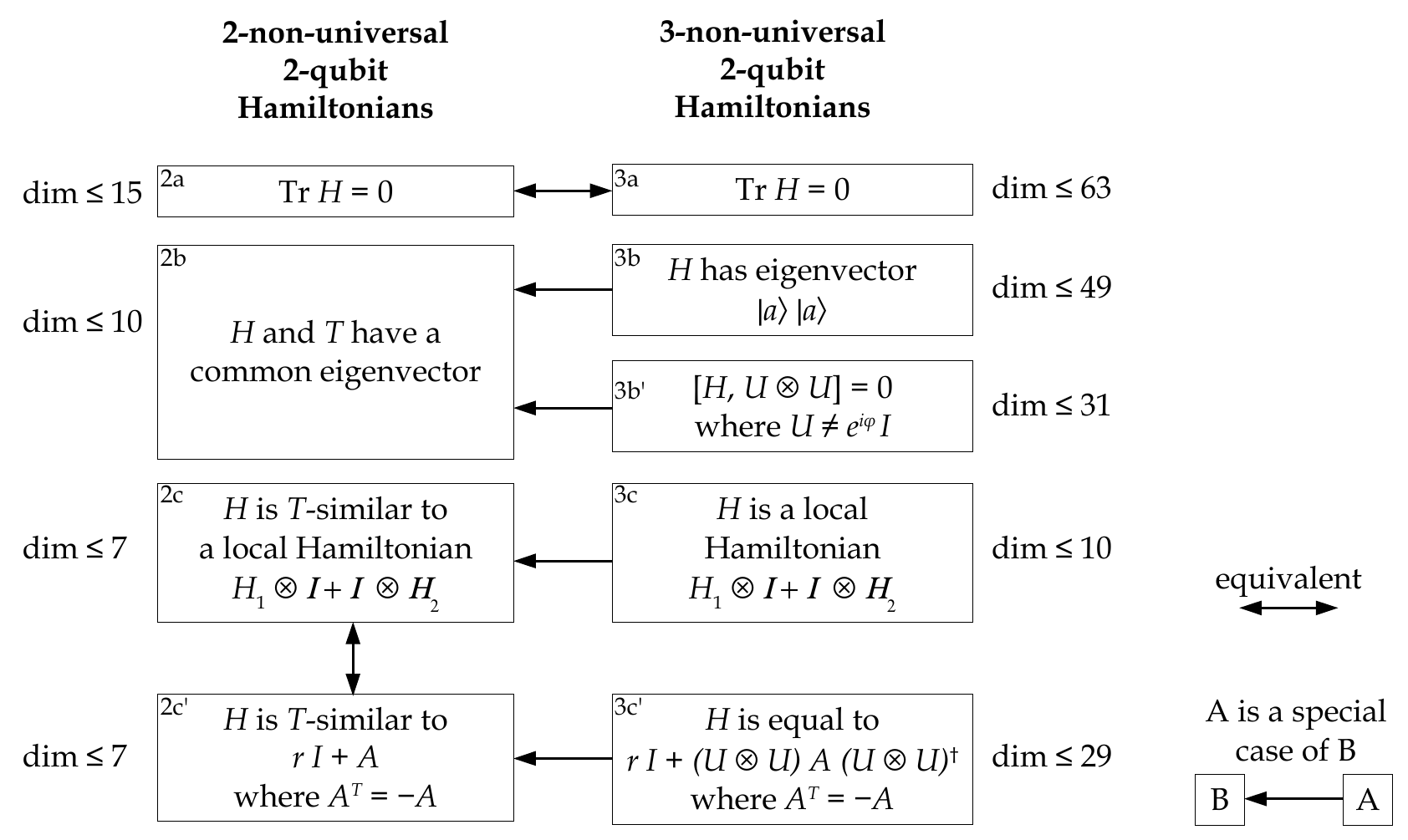}
  \caption[Relations between \nonuniversal{2} and \nonuniversal{3} Hamiltonians]{Relations between the families of \nonuniversal{2} and \nonuniversal{3} Hamiltonians. For each family $\mc{F}$ we give the maximum achievable value of $\dim \Lie{H,THT}$ for $H \in \mc{F}$.}
  \label{fig:Diagram}
\end{figure}

\section{Open problems}
\label{sec:Conclusion}

The main result of this paper is a complete characterization of \universal{2} \qubit{2} Hamiltonians, as summarized in Theorem \ref{thm:2NonUniversal}. Several variants of the problem that remain open:
\begin{enumerate}
  \item Which \nonuniversal{2} \qubit{2} Hamiltonians are \universal{n}, i.e., become universal on $n \geq 3$ qubits? Also, is there $n_0 \in \N$ such that \mbox{$n_0$-non-universality} implies \mbox{$n$-non-universality} for all $n \geq n_0$? In particular, is $n_0 = 3$?
  \item Which \qubit{2} Hamiltonians are universal with ancillae? (See Section~\ref{sec:Ancillae} and Definition \ref{def:UniversalAncilla} for the definition of universality with ancillae for unitary gates.) One might also consider a scenario in which the number of allowed ancillary qubits is restricted.
  \item Which \qubit{2} Hamiltonians give us encoded universality, e.g., generate \O{4}? This question is relevant since the full power of quantum computation can be achieved even with a restricted repertoire of gates. For example, real gates are sufficient \cite{Rebit,Bernstein} since \O{2 \cdot 2^n} contains \U{2^n}. We can say that $H$ is \emph{\universal{n} in an encoded sense} if there exists $k \in \N$ (possibly depending on $n$) such that the Lie algebra generated by $H$ on $n+k$ qubits, $\mc{L} \subseteq \u{2^{n+k}}$, contains \u{2^n} as a subalgebra. However, it is not even clear how to check this for a \emph{particular} Hamiltonian.
\end{enumerate}

\section*{Acknowledgments}

We thank John Watrous for helpful discussions. Support for this research was provided by MITACS, NSERC, QuantumWorks, and the US ARO/DTO. DL also acknowledges support from CRC, CFI, ORF, and CIFAR.

\appendix

\section{Positive time evolution is sufficient}
\label{app:NegativeTime}

\begin{claim}
\label{cl:NegativeTime}
Let $H\in\u{N}$ be a Hamiltonian and let $\tau<0$. Then for all $\varepsilon>0$ there exists $t > 0$ such that $\opnorm{\eto{H\tau}-\eto{Ht}}<\varepsilon$.
\end{claim}

\begin{proof}
Let $U\defeq\eto{H}$. Consider the sequence $\mc{K}\defeq\set{U^i}_{i=1}^{\infty}\subset\M{N}{\C}$. Note that we can think of $\M{N}{\C}$ as a real vector space of dimension $2N^2$. Since $\mc{K}$ is bounded with respect to the spectral norm, by the Bolzano-Weierstrass theorem, $\mc{K}$ has a convergent subsequence. It follows that for all $\varepsilon>0$ and all $n_0\in\N$ there exist $j,k\in\N$ such that $j-k>n_0$ and $\varepsilon>\opnorm{U^k-U^j}=\opnorm{\id_N-U^{j-k}}$. Equivalently, for all $\varepsilon>0$ and all $n_0\in\N$ there exists $n>n_0$ such that $\varepsilon > \opnorm{\id_N-U^n}$. Therefore, given $\tau<0$, for all $\varepsilon>0$ there exists $n >\abs{\tau}$ such that
\begin{equation}
  \varepsilon>\opnorm{\id_N-\eto{Hn}}=\opnorm{\eto{H\tau}-\eto{H(n+\tau)}}.
\end{equation}
Taking $t\defeq n+\tau>0$, the claim follows.
\end{proof}

\section{Basic properties of the \texorpdfstring{$T$}{T} gate}
\label{app:Tgate}

In this appendix we restate and prove the basic properites of the $T$ gate introduced in Section \ref{sec:T}.

\CommutesWithTFact*

\begin{proof}
Suppose $\ket{s}$ is an eigenvector of $N$. Then $\mc{B}=\set{\ket{s},\ket{n_1},\ket{n_2},\ket{n_3}}$ is an orthonormal eigenbasis of $N$ for some orthonormal vectors $\set{\ket{n_i}}_{i=1}^3\subset\C^4$. Since $\mc{B}$ is orthonormal, $\set{\ket{n_i}}_{i=1}^3\in E_-^{\perp}=E_+$. Therefore, $\mc{B}$ is also an eigenbasis of $T$, and both $N$ and $T$ are simultaneously diagonal in this basis. Thus $[N,T]=0$.

Conversely, suppose $[N,T]=0$. Then $N$ and $T$ are simultaneously diagonal in some orthonormal basis $\mc{B}$. Since $\ket{s}$ spans the one-dimensional eigenspace $E_-$ of $T$, we know that $e^{i\phi}\ket{s}\in\mc{B}$ for some $\phi\in\R$. Thus, $\ket{s}$ is an eigenvector of $N$.
\end{proof}

\OrthogonalEigenvectorFact*

\begin{proof}
The ``if'' direction is trivial. Conversely, suppose $N$ shares an eigenvector $\ket{v}$ with the $T$ gate. If $\ket{v}\in E_+$ we are done. Otherwise, $\ket{v}\in E_-$, so $\ket{v}=\ket{s}$. Then $\spn_{\C}(\ket{v})^\perp = E_+$ is an invariant subspace of $N$ and it contains an eigenvector of $N$.
\end{proof}

\SingletEigenvecFact*

\begin{proof}
Since $[U,T]=0$, we know that $U$ and $T$ are simultaneously diagonal in some orthonormal basis. The singlet $\ket{s}$ must belong to this basis, since it spans the \mbox{one-dimensional} eigenspace $E_-$ of the $T$ gate. Therefore, $\ket{s}$ has to be an eigenvector of $U$ as well. Note that $U$ and $U\ct$ have the same eigenvectors. Thus, $\ket{s}$ is also an eigenvector of $U\ct$.
\end{proof}

\begin{fact}
\label{fact:DegEigenvals}
If a \qubit2{} Hamiltonian $H$ has a degenerate eigenvalue, then it shares an eigenvector with $T$ and hence is not universal.
\end{fact}

\begin{proof}
Suppose $H$ has a degenerate eigenvalue, and let $E$ denote the corresponding eigenspace. Recall that the $T$ gate has a \mbox{3-dimensional} \mbox{$(+1)$-eigenspace} $E_+$. Now note that the intersection $E \cap E_+$ is at least \mbox{1-dimensional}, since $E,E_+\subseteq \C^4$ and $\dim(E)\geq 2$, $\dim(E_+)=3$. Any nonzero $\ket{v}\in E \cap E_+$ is a common eigenvector of $H$ and the $T$ gate. By Fact~\ref{fact:Hyp1} we conclude that $H$ is \mbox{non-universal}.
\end{proof}

\section{\Tsimilar{ity} to a local Hamiltonian}
\label{app:Tsimilarlocal}

In this appendix we prove a result characterizing Hamiltonians that are $T$-similar to some local Hamiltonian, as stated in Section \ref{sec:TsimTransformations}. Our proof makes use of the following general characterization of \Tsimilar{}ity:

\begin{theorem}
\label{thm:Tsimilarity}
Hamiltonians $H$ and $H'$ are \Tsimilar{} if and only if there exist orthonormal eigenbases $\{\ket{v_1},\ket{v_2},\ket{v_3},\ket{v_4}\}$ of $H$ and $\{\ket{w_1},\ket{w_2},\ket{w_3},\ket{w_4}\}$ of $H'$ such that $\bra{v_i}H\ket{v_i}=\bra{w_i}H'\ket{w_i}$ and $|{\braket{v_i}{s}}|=|{\braket{w_i}{s}}|$ for all $i \in \{1,2,3,4\}$,
where $\ket{s}$ is the singlet state defined in equation (\ref{eq:Singlet}).
\end{theorem}

\begin{proof}
Assume $H$ and $H'$ are \Tsimilar{}, i.e., $H'=U H U\ct$ for some $U\in\uF$ with $[U,T]=0$. Since $[U,T]=0$, by Fact~\ref{fact:SingletEigenvec} we know that $\ket{s}$ is an eigenvector of $U\ct$. Let $\ket{v}$ be an eigenvector of $H$. Then $U\ket{v}$ is the corresponding eigenvector of $U H U\ct$. Now we have $\abs{\bra{s}(U\ket{v})}=\abs{(U\ct\ket{s})\ct\ket{v}} = \abs{\braket{s}{v}}$, i.e., the corresponding eigenvectors of $H$ and $U H U\ct$ have the same overlaps with the singlet state. Since conjugation does not change the eigenvalues, the ``only if'' direction of the theorem follows.

Conversely, assume that $H$ and $H'$ have orthonormal eigenbases $\{\ket{v_i}\}$ and $\{\ket{w_i}\}$, respectively, with $\lambda_i \defeq \bra{v_i}H\ket{v_i}=\bra{w_i}H'\ket{w_i}$ and $r_i \defeq |{\braket{v_i}{s}}|=|{\braket{w_i}{s}}|$. We can express the singlet state $\ket{s}$ in the eigenbases of $H$ and $H'$ as follows:
\begin{align}
  \ket{s} &= \sum_{j=1}^4 r_j e^{i\alpha_j} \ket{v_j}\label{eq:sv} \\
          &= \sum_{j=1}^4 r_j e^{i\beta_j} \ket{w_j},\label{eq:sw}
\end{align}
where $\alpha_j,\beta_j\in\R$. Now let
\begin{equation}
  U\defeq\sum_{j=1}^4 e^{i(\beta_j-\alpha_j)}\ket{w_j}\bra{v_j}.
\label{eq:U}
\end{equation}

We claim that (a) $U H U\ct = H'$ and (b) $\abs{\bra{s}U \ket{s}} = 1$.
\begin{enumerate}[(a)]
  \item By expressing $U$ as in (\ref{eq:U}), we have
    \begin{align}
      U H U\ct = & \sum_{j=1}^4 e^{i(\beta_j-\alpha_j)}\ket{w_j}\bra{v_j} \sum_{k=1}^4 \lambda_k \ket{v_k}\bra{v_k}
                   \sum_{l=1}^4 e^{-i(\beta_l-\alpha_l)}\ket{v_l}\bra{w_l}\\
               = & \sum_{k=1}^4 e^{i(\beta_k-\alpha_k)} \lambda_k e^{-i(\beta_k-\alpha_k)}\ket{w_k}\bra{w_k}
               =   \sum_{k=1}^4 \lambda_k \ket{w_k}\bra{w_k} = H'.
    \end{align}
  \item By expressing $\bra{s}$ as in (\ref{eq:sw}), $\ket{s}$ as in (\ref{eq:sv}), and $U$ as in (\ref{eq:U}), we get
    \begin{align}
      \bra{s} U \ket{s} =  & \sum_{j=1}^4 r_j e^{-i\beta_j}\bra{w_j} \sum_{k=1}^4 e^{i(\beta_k-\alpha_k)}\ket{w_k}\bra{v_k}
                                                                     \sum_{l=1}^4 r_l e^{i\alpha_l}\ket{v_l} \\
                        = & \sum_{k=1}^4 r_k e^{-i\beta_k} e^{i(\beta_k-\alpha_k)}  r_k e^{i\alpha_k} = \sum_{k=1}^4 r_k^2=1.
    \end{align}
\end{enumerate}

Part (a) tells us that $H$ and $H'$ are similar via $U$. From (b) it follows that $\ket{s}$ is an eigenvector of $U$, so by Fact~\ref{fact:CommutesWithT}, $U$ commutes with $T$. Hence $H$ and $H'$ are \Tsimilar{}.
\end{proof}

Now we can prove the result stated in Section \ref{sec:TsimTransformations} characterizing \Tsimilar{}ity to a local Hamiltonian:

\PatternLemma*

\begin{proof}
First we prove the ``only if'' direction. Suppose $H$ is \Tsimilar{} to some local Hamiltonian $H' = H_1 \otimes \id + \id \otimes H_2$. By Theorem~\ref{thm:Tsimilarity}, it suffices to show that the spectrum of $H'$ has the form described in the lemma. We diagonalize $H_1$ and $H_2$ as follows:
\begin{equation}
  H_1 = \alpha_1 \ket{v_1} \bra{v_1} + \alpha_2 \ket{v_2} \bra{v_2}, \quad
  H_2 = \beta_1  \ket{w_1} \bra{w_1} + \beta_2  \ket{w_2} \bra{w_2}.
\end{equation}
Let the first eigenvectors of $H_1$ and $H_2$ be
\begin{equation}
  \ket{v_1} = \mx{a\\b}, \quad
  \ket{w_1} = \mx{c\\d}.
  \label{eq:v1w1}
\end{equation}
Since we can ignore the global phase of each eigenvector, we may assume that
\begin{equation}
  \ket{v_2} = \mx{-b\co \\ a\co}, \quad
  \ket{w_2} = \mx{-d\co \\ c\co}.
  \label{eq:v2w2}
\end{equation}
Then
\begin{equation}
  \ket{v_1} \otimes \ket{w_1}, \quad
  \ket{v_2} \otimes \ket{w_2}, \quad
  \ket{v_1} \otimes \ket{w_2}, \quad
  \ket{v_2} \otimes \ket{w_1}
  \label{eq:TensorEigenvectors}
\end{equation}
are eigenvectors of $H'$. Calculating the overlaps with $\ket{s}$, we find
\begin{align}
  \abs{\braket{s}{v_1,w_1}}^2 & = \tfrac{1}{2} \abs{a d - b c}^2 \qefed r, \label{eq:PseudoScalarProduct} \\
  \abs{\braket{s}{v_2,w_2}}^2 & = \tfrac{1}{2} \abs{a\co d\co - b\co c\co}^2 = \tfrac{1}{2} \abs{a d - b c}^2 = r, \\
  \abs{\braket{s}{v_1,w_2}}^2 & = \tfrac{1}{2} \abs{a c\co + b d\co}^2 \qefed t, \label{eq:ScalarProduct} \\
  \abs{\braket{s}{v_2,w_1}}^2 & = \tfrac{1}{2} \abs{-a\co c - b\co d}^2 = \tfrac{1}{2} \abs{a c\co + b d\co}^2 = t.
\end{align}
The corresponding eigenvalues of $H'$ are
\begin{equation}
  \lambda_1 = \alpha_1 + \beta_1, \quad
  \lambda_2 = \alpha_2 + \beta_2, \quad
  \lambda_3 = \alpha_1 + \beta_2, \quad
  \lambda_4 = \alpha_2 + \beta_1,
\end{equation}
respectively; they satisfy $\lambda_1 + \lambda_2 = \lambda_3 + \lambda_4$. This establishes the ``only if'' direction.

Now let us prove the ``if'' direction. For any $H$ with a spectrum satisfying the conditions of the lemma, we construct a local Hamiltonian $H'=H_1\otimes\id+\id\otimes H_2$ that is \Tsimilar{} to $H$. As before, let $\alpha_i$ and $\ket{v_i}$ denote corresponding eigenvalues and eigenvectors of $H_1$, and let $\beta_i$ and $\ket{w_i}$ denote eigenvalues and eigenvectors of $H_2$. In terms of the eigenvalues $\lambda_i$ of $H$, we choose the eigenvalues of $H_1$ and $H_2$ as follows: $\alpha_1 = 0$, $\alpha_2 = \lambda_2-\lambda_3$, $\beta_1 = \lambda_1$, and $\beta_2 = \lambda_3$. With this choice, the eigenvalues of $H'$ are
\begin{equation}
  \alpha_1 + \beta_1 = \lambda_1, \quad
  \alpha_2 + \beta_2 = \lambda_2, \quad
  \alpha_1 + \beta_2 = \lambda_3, \quad
  \alpha_2 + \beta_1 = \lambda_4,
\end{equation}
where the last equality holds since $\lambda_1+\lambda_2=\lambda_3+\lambda_4$.
It remains to choose eigenvectors of $H_1$ and $H_2$ to obtain the required overlaps with $\ket{s}$. Notice that $\ket{v_1}$ and $\ket{w_1}$ completely determine the overlaps, since without loss of generality we can take $\ket{v_2}$ and $\ket{w_2}$ as in (\ref{eq:v2w2}). In fact, it suffices to choose $\ket{v_1}, \ket{w_1} \in \R^2$. If the angle between real unit vectors $\ket{v_1}=\smx{a\\b}$ and $\ket{w_1}=\smx{c\\d}$ is $\theta$, then $a d - b c = \sin{\theta}$ and $a c + b d = \cos{\theta}$. Thus, the overlaps (\ref{eq:PseudoScalarProduct}) and (\ref{eq:ScalarProduct}) are $\frac{1}{2} \sin^2 \theta = r$ and $\frac{1}{2} \cos^2 \theta = t$, respectively. Therefore, we can take any two real unit vectors having angle $\theta = \arcsin \sqrt{2 r}$. Since $H$ and $H'$ satisfy the conditions of Theorem \ref{thm:Tsimilarity}, they are \Tsimilar.
\end{proof}

\section{Almost all two-qubit unitaries are \universal{2}}
\label{app:UnitaryUniversality}

In this paper, we have primarily focused on universality of Hamiltonians.  In this appendix, we return to the question of universality of unitary gates.  In particular, we show how our results easily imply that a Haar-uniform \qubit{2} unitary is almost surely universal.

\begin{definition}
\label{def:SimulateU}
We say that a \qubit{2} unitary $U$ is \emph{\universal{2}} if, for any $V\in \U{4}$ and any $\varepsilon>0$, there exist $l \in \N$ and $N_1, \dotsc, N_l \in\N \cup \set{0}$ such that
\begin{equation}
  \opnorm{V-U^{N_l} \dotso U^{N_3} (TUT)^{N_2} U^{N_1}} < \varepsilon.
\end{equation}
\end{definition}

We say that a Hamiltonian $H$ \emph{corresponds} to a unitary $U$ if $e^{-iH} = U$. In general, for a given $U$, such Hamiltonian is not unique (if $H$ corresponds to $U$ then so does $H + 2 \pi k \ket{v}\bra{v}$ for any integer $k$ and any eigenvector $\ket{v}$ of $H$). However, if we also demand that all eigenvalues of $H$ are in the interval $(-\pi,\pi]$, then there is a unique such Hamiltonian.

If $\lambda_1, \dotsc, \lambda_4$ are the eigenvalues of $U \in \U{4}$, we can associate to each $\lambda_j$ a unique phase $\theta_j \in (-\pi,\pi]$ such that
\begin{equation}
  \lambda_j = e^{-i\theta_j}.
\end{equation}
We call these the \emph{canonical phases of $U$}. Furthermore, if the eigenvalues $\lambda_1, \dotsc, \lambda_4$ of $U$ are distinct, we can associate to each $\lambda_j$ a unique eigenvector $\ket{v_j}$ such that $U = \sum_{j=1}^4 \lambda_j \ket{v_j}\bra{v_j}$. We call
\begin{equation}
  H := \sum_{j=1}^4 \theta_j \ket{v_j}\bra{v_j}
\end{equation}
the \emph{canonical Hamiltonian} corresponding to $U$. Note from Lemma~\ref{lem:RatIndep} below that, for a Haar-random $U$, all $\theta_j$ are almost surely distinct, in which case there is no freedom in choosing the eigenvectors $\ket{v_j}$.

Building on our characterization of \universal{2} Hamiltonians, we now show that a Haar-random unitary matrix $U\in\U{4}$ is almost surely \universal{2}. The main idea is that the canonical Hamiltonian $H$ corresponding to such $U$ is almost surely \universal{2} and, moreover, natural powers of $U$ can be used to approximate the evolution according to $H$ for any amount of time. To achieve the latter, we need to ensure that the eigenvalues of $H$ are rationally independent with probability one. Recall that numbers $\alpha_1, \dotsc, \alpha_n \in \R$ are \emph{rationally dependent} if $\sum_{i=1}^n q_i \alpha_i = 0$ for some $q_1, \dotsc, q_n \in \Q$, not all of which are zero.

\begin{lemma}
Let $\thetas \in (-\pi,\pi]$ be the canonical phases of a Haar-random unitary $U\in\U{4}$. Then with probability one the angles $\pi, \thetas$ are rationally independent.
\label{lem:RatIndep}
\end{lemma}

\begin{proof}
The joint density of eigenvalues of a Haar-random unitary was obtained by Weyl and is given by 
\begin{equation}
  f(\thetas) := \frac{1}{(2\pi)^4 4!}
  \prod_{j<k}\abs{e^{i\theta_j}- e^{i\theta_k}}^2
  \label{eq:f}
\end{equation}
(see \textit{e.g.} \cite{Density}).
For the purpose of this proof the specific form of $f(\thetas)$ is irrelevant; we will only use the fact that $f(\thetas)$ is upper bounded by some fixed finite constant.

Consider the set $(-\pi,\pi]^4$ of all possible canonical phases of a $4 \times 4$ unitary. Let us exclude from this set all those points $\vtheta := (\thetas)$ for which $\pi, \thetas$ are rationally dependent, i.e., $q_0 \pi + \sum_{i=1}^4 q_i \theta_i = 0$ for some $q_0, q_1, \dotsc, q_4 \in \Q$. To account for all possible rational dependences, for each choice of coefficients $\vq := (q_0, q_1, \dotsc, q_4) \in \Q^5$, we define an affine hyperplane
\begin{equation}
  A_{\vq} := \set{\vtheta \in \R^4 : q_0 \pi + \sum_{i=1}^4 q_i \theta_i = 0}
\end{equation}
consisting of all those angles $\vtheta$ that are rationally dependent and their dependence can be expressed with coefficients given by $\vq$. Note that the set $A_{\vq}$ has zero measure in $\R^4$, so
\begin{equation}
  A := \bigcup_{\vq \in \Q^5} A_{\vq}
\end{equation}
also has zero measure, as it is a countable union of measure-zero sets.

If $S := (-\pi,\pi]^4$ is the set of all canonical phases and $\mu$ denotes the standard Lebesgue measure on $\R^4$, the probability that a Haar-random unitary has rationally independent canonical phases is
\begin{equation}
  \int_{S \setminus A} f d\mu,
\end{equation}
where $f$ is given by equation \eqref{eq:f}. However, since the set $S \setminus A$ differs from $S$ only by a set of measure zero, the Lebesgue integral of a bounded function over the two sets is the same \cite[p.~40]{Cramer}. Thus,
\begin{equation}
  \int_{S \setminus A} f d\mu
  =  \int_S f d\mu
  = 1
\end{equation}
and the result follows.
\end{proof}

\begin{lemma}
Let $\vtheta = (\thetasn)\in (-\pi,\pi]^n$ be an $n$-tuple of angles, where each $\theta_i$ is irrational with respect to $\pi$. Then for any $\eps > 0$ there exists $N\in\N$ such that $0<\abs{N \theta_i \bmod (-\pi,\pi]} \le \eps$ for all $i \in \{1,\dotsc,n\}$.
\label{lem:SmallRot}
\end{lemma}

\begin{proof}
Pick $m := \lceil 2\pi/\eps \rceil$ and partition $(-\pi,\pi]$ into $m$ segments of size $2\pi/m \le \eps$.
For a given $r\in\N$, the components of $r\vtheta$ fall into a combination of these segments. Each such combination can be specified by an element of
the finite set $\set{1,\dotsc,m}^n$.
On the other hand, the number of different possible values of $r\in\N$ is infinite. Thus, by the pigeonhole principle, there exist distinct $r,s\in\N$ such that for each $i$, the values $r\theta_i$ and $s\theta_i$ modulo $2\pi$ fall in the same segment (which can depend on $i$).   
For each $i$, by the irrationality assumption, $r \theta_i \not\equiv s\theta_i \bmod 2\pi$; moreover, both $r\theta_i$ and $s\theta_i$ fall in the interior of the $(2\pi/m)$-sized segments. Thus, taking $N := \abs{r-s}$ ensures that the components of $N \vtheta \bmod (-\pi,\pi]$ are all $\eps$-close to yet distinct from zero.
\end{proof}

We are now ready to show that for almost all unitary matrices $U$, we can use their natural powers to simulate the evolution $\eto{Ht}$ according to the canonical Hamiltonian $H$ for any time $t\in \R$.
\begin{lemma}
Let $U\in\U{4}$ and let $\vtheta=(\thetas)\in(-\pi,\pi]^4$ be the canonical phases of $U$. If the values $\thetas,\pi$ are rationally independent then for any $t \in \R$ and any $\eps>0$ there exists $N\in\N$ such that
\be
  \norm{U^{N} - \eto{Ht}}_\infty < \eps,
\ee
where $H$ is the canonical Hamiltonian corresponding to $U$.
\label{lem:GoodU}
\end{lemma}

\begin{proof}
Note that
\be
  \norm{U^{N} - \eto{Ht}}_\infty \le \sum_{j=1}^{4} 
  \abs{e^{i N \theta_j} - e^{-i t \theta_j}} .
\label{eq:Simul}
\ee
Now let $\vphi = (\phi_1,\phi_2,\phi_3,\phi_4) := - t \vtheta \bmod (-\pi,\pi] \in (-\pi,\pi]^4$ and note that
\be
  \abs{e^{i N \theta_j} - e^{i \phi_j}}
  = \bigl|e^{i (N \theta_j-\phi_j)} - 1\bigr|
  = \sqrt{2-2\cos(N\theta_j - \phi_j)}
  \le \abs{N \theta_j - \phi_j}.
\label{eq:BoundDiff}
\ee
Combining equation~(\ref{eq:Simul}) and equation~(\ref{eq:BoundDiff}), we obtain
\be
  \norm{U^{N} - \eto{Ht}}_\infty \le \| N \vtheta - \vphi \|_1.
\ee
Our goal is to pick $N$ so that $\| N \vtheta - \vphi \|_1 \le \eps$.

By Lemma~\ref{lem:SmallRot}, we can find $M\in \N$ such that each of the components of $\vtheta' := M \vtheta \bmod (-\pi,\pi]$ is $\frac{\eps}{8}$-close to zero. Next, let us consider the linear flow $\Phi(\vtheta',x) := x\vtheta' \bmod (-\pi,\pi]$ on the 4-torus $(-\pi,\pi]^4$. Since the values $\theta'_1,\theta'_2,\theta'_3,\theta'_4, \pi$ are rationally independent, the flow $\Phi(\vtheta',x)$ is dense in~$(-\pi,\pi]^4$ \cite{Flows}. Thus we can find $x\in\R$ such that $\|\vphi-\Phi(\vtheta',x)\|_1 < \frac{\eps}{2}$. Furthermore, we have $\bigl\| \Phi(\vtheta',x) - \Phi(\vtheta',\lfloor x \rfloor) \bigr\|_1 = 
\bigl\| \vtheta'\bigr\|_1 \bigl|x-\lfloor x \rfloor\bigr|< \frac{\eps}{2}$, since the four components of $\vtheta'$ are all $\frac{\eps}{8}$-close to zero. 
Combining the last two inequalities yields $\bigl\| \vphi - \Phi(\vtheta',\lfloor x \rfloor) \bigr\|_1 < \eps$. Therefore, taking $N := M \lfloor x \rfloor$ ensures that $\| N\vtheta - \vphi \|_1 < \eps$ and hence $\norm{U^N  - e^{-iHt}}_\infty < \eps$ as desired.
\end{proof}

We are now ready to show that almost all \qubit{2} unitaries are universal.

\begin{theorem}
A Haar-random unitary $U \in \U{4}$ is almost surely \universal{2}.
\end{theorem}
\begin{proof}
By Lemma~\ref{lem:RatIndep}, the canonical phases $\thetas\in(-\pi,\pi]$ of $U$ are almost surely rationally independent. Hence, it suffices to prove the theorem statement for such unitaries $U$.

Applying Lemma~\ref{lem:GoodU} to such unitaries $U$, 
we can use natural powers of $U$ to simulate the evolution according to its corresponding canonical Hamiltonian $H$ for any time $t\in\R$. So to establish the theorem it remains to argue that $H$ is  \universal{2}. We do this by showing that $H$ does not satisfy any of the three conditions in Theorem~\ref{thm:2NonUniversal}. If $H$ is $T$-similar to a local Hamiltonian then $\theta_i+\theta_j = \theta_k +\theta_l$ for $\set{i,j,k,l}=\set{1,\dotsc,4}$. Similarly, if $\tr(H)=0$ then $\sum_{i=1}^4 \theta_i = 0$. Therefore, each of the conditions (1) and (3) implies that the $\theta_i$ are rationally dependent, which contradicts our initial assumption. Let us now examine condition~(2). If $H$ shares an eigenvector with $T$ then, by Fact~\ref{fact:OrthogonalEigenvector}, $H$ (and hence also $U$) has an eigenvector orthogonal to the singlet state $\ket{s} = \frac{1}{\sqrt{2}}(\ket{01}-\ket{10})$. The probability that a Haar-random $U$ has an eigenvector orthogonal to $\ket{s}$ is zero. Therefore, $H$ almost surely does not satisfy any of the three conditions and is thus \universal{2}.
\end{proof}

This provides an alternate proof of the main result of \cite{DBE,Lloyd}, avoiding the shortcomings discussed in Section~\ref{sec:PreviousResultsSingleGate}.  However, unlike in our main result on Hamiltonian universality (Theorem \ref{thm:2NonUniversal}), we have not characterized precisely which \qubit{2} unitaries are universal.  To the best of our knowledge, it remains open to find such a characterization.


\end{document}